\documentclass[11pt]{article}
\usepackage[margin=1in]{geometry} 
\usepackage{amsmath}
\usepackage{amscd}
\usepackage{amssymb,amsfonts,bbm,mathrsfs,stmaryrd}

\usepackage[thmmarks,hyperref]{ntheorem}
\usepackage{url}
\usepackage{authblk}

\theoremstyle{change}

\newtheorem{definition}[equation]{Definition}

\newtheorem{theorem}[equation]{Theorem}

\newtheorem{proposition}[equation]{Proposition}

\newtheorem{cor}[equation]{Corollary}
\newtheorem{conjecture}[equation]{Conjecture}

\newtheorem{example}[equation]{Example}

\theorembodyfont{\upshape}
\theoremsymbol{\ensuremath{\Diamond}}

\newtheorem{remark}[equation]{Remark}

\theoremstyle{nonumberplain}

\theoremsymbol{\ensuremath{\Box}}
\newtheorem{proof}{Proof}

\qedsymbol{\ensuremath{\Box}}
\numberwithin{equation}{section}

\usepackage{float}

\newcommand\ket[1]{\mid #1 \rangle}
\newcommand\setof[1]{\{ #1 \}}
\newcommand\lt{<}
\newcommand\abs[1]{ \mid #1 \mid }

\title{G-Extensions of Quantum Group Categories\\ and Functorial SPT}
\author{Ammar Husain \thanks{Electronic address: \texttt{ahusain@berkeley.edu}}}
\affil{UC Berkeley}
\date{}

\begin{document}
\maketitle

\abstract

In this short mostly expository note, we sketch a program for gauging fully extended topological field theories in 3 dimensions. One begins with the spherical fusion category with which one wants to do Levin-Wen or Turaev-Viro. One then computes a homotopic space with certain $\pi_\bullet$ given by autoequivalences, invertible objects and the ground (algebraic) field. Then for each desired symmetry group $G$ one looks at mapping $BG$ into that. This classifies equivalence classes of G-extended fusion categories. This is an equivalence at a fully extended level so will allow many defects rather than only evaluating partition functions on closed manifolds. We can now use these categories to build a new fully extended 3d topological field theory, but now possibly with extra data. This is the result of permeating defect walls and saying how that affects the assignment to the point strata.

\section{Introduction}

In recent years, there has been a focus on symmetry protected topological phases \cite{FidkowskiAshvin,Barkeshli,RyanCurt,WenGroupCoh}. One of the questions there is to further gauge a finite symmetry group. This can be phrased in terms of the cobordism theorem \cite{Baez,Lurie} in the case of three dimensional theories. In this framework, some obstructions become more manifest. Physically this is by viewing all symmetries as codimension 1 defects which implement that symmetry for everything in the theory at once including all the defects. This comes from understanding the action on the data assigned to the point instead of only knowing the action on the Hilbert space. If the theory was determined by a spherical fusion category assigned to the point, then this procedure of gauging is understood in the context of \cite{ENO10} where we must understand the $(\infty , 3)$ Morita category. It can then be asked what kind of field theory we get if the result is then assigned to the point. By being a fully extended theory, this model can be evaluated on stratified spaces with many defects. This is especially interesting in the context of quantum groups where many examples of fusion categories can be readily produced.

This is a mostly expository note that grew out of some email exchanges in November 2015. We will first review the way that Turaev-Viro models fit into the cobordism framework. Then how the Levin-Wen model gives a particular lattice realization. We then give the definitions of G-extensions of \cite{ENO10}. The next sections of equivariantization and transmutation give relations to other methods in the literature. We then say how this translates to giving anomaly inflow followed by some sample calculations in the case of some elementary quantum group categories.

\begin{remark}
There is some abuse of notation by calling all of these Turaev-Viro theories when those should really be the spherical case. See the adjectives framed and combed to distinguish.
\end{remark}

\section{The essential bit of $\infty$ Categories}

\begin{theorem}[Cobordism Hypothesis]\cite{Lurie}
n-dimensional local framed topological field theories with target a symmetric monoidal $(\infty , n)$ category $\mathcal{C}$ are in one to one correspondence with the fully dualizable objects of $\mathcal{C}$. In fact, the space of such field theories is homotopy equivalent to the space of n-dualizable objects of $\mathcal{C}$.
\end{theorem}

\begin{definition}[$MonCat_{bim}$ \cite{DSPS}]
This is the $(\infty , 3 )$ category, whose objects are finite rigid monoidal linear categories, morphisms are bimodule categories between these, 2-morphisms are bimodule functors, 3-morphisms natural transformations and from there on are equivalences.
Fusion categories are fully dualizable objects in this $(\infty , 3)$ category.
\end{definition}

\begin{theorem}[\cite{DSPS}]
For any seperable tensor category, there is a 3-dimensional 3-framed local topological field theory whose value on a point is that tensor category. In particular, there is such a field theory for any finite semisimple tensor category over a field of characteristic zero, and such a field theory for any fusion category of nonzero global dimension over an algebraically closed field of finite characteristic. 
\end{theorem}

\begin{remark}
The 3-framing can be thought of as a vielbien for more familiar language.
\end{remark}

\begin{conjecture}[\cite{DSPS}]
Every pivotal fusion category in characteristic zero admits the structure of an $SO(2)$ homotopy fixed point, and therefore provides the structure of a combed 3-dimensional local field theory.
\end{conjecture}

\begin{conjecture}[\cite{DSPS}]
Every spherical fusion category in characteristic zero admits the structure of an $SO(3)$ homotopy fixed point, and therefore provides the structure of an oriented 3-dimensional local field theory.
\end{conjecture}

So to apply a symmetry transformation $g \in G$ we must apply a change of the fusion category and it's fully dualization data. This is given by a $\mathcal{C} - \mathcal{C}$ bimodule and higher coherences. This is exactly what we will see momentarily.

\section{Levin Wen}

The Levin Wen model is built from an arbitrary unitary spherical fusion category $\mathcal{C}$\cite{LevinWen}. From this data a lattice model is built such that the bulk excitations are given by simple objects of the Drinfeld center $Z(\mathcal{C})$. The model is defined on a trivalent planar graph. The Hilbert space assigned to this graph is given by the finite dimensional Hilbert space spanned by the basis vectors which are labellings of edges by simple objects of the category and vertices by a basis vector of the multiplicity space for the fusion of the two incoming edges into the outgoing edge. For other orientations, this is corrected by dualizing and reversing arrows. The Hamiltonian is then cooked up with projectors. \cite{KitaevKong}

\begin{eqnarray*}
H &=& \sum_v ( 1 - Q_v ) + \sum_p ( 1 - B_p )\\
B_p &=& \sum_{k \in I} \frac{d_k}{D^2} B_p^k\\
D^2 &=& \sum_{i \in I} d_i^2
\end{eqnarray*}

\begin{definition}
The Drinfeld Center of a monoidal category $\mathcal{C}$ is the monoidal category of endo-pseudonatural transformations of the identity 2-functor on $B \mathcal{C}$ which comes from the identity functor on $\mathcal{C}$. This is because a natural transformation needs to give first of all an isomorphism from $id( pt) \to id (pt)$. That is given by a morphism of $B \mathcal{C}$ aka an object X of $\mathcal{C}$. And also the squares need to have a filling 2-morphism $\phi_Y$ a.k.a a morphism in $\mathcal{C}$ that does $X \otimes Y \to Y \otimes X$ which has to be natural in Y.

\end{definition}

If the previous section was used to define a $2+1$ dimensional field theory, then the particles in the form of constraint violations would be given by $Z(\mathcal{C})$ because the excitations are assigned to the circles. This is exactly what is seen in this model. This is where the fact that the center of spherical gives modular has appeared. For further details on this relationship see \cite{Balsam4} and note that other lattice models can be used as realizations as well as taking low energy limits.

\section{Brauer-Picard}

\begin{definition}[$BiModc( \mathcal{C} )$ ]
For a fusion category $\mathcal{C}$, we have the monoidal 2-category of $\mathcal{C}$ bimodule categories. The 1-morphisms are functors of bimodule categories and the 2-morphisms are natural transformations of such functors. The tensor product is given by $\boxtimes_\mathcal{C}$ which is defined through the universal property $Fun_{bal} ( \mathcal{M} \times \mathcal{N} , \mathcal{A}) \simeq Fun( \mathcal{M} \boxtimes_{\mathcal{C}} \mathcal{N} , \mathcal{A} )$ to all abelian categories $\mathcal{A}$.
\end{definition}

\begin{definition}[Bimodc]
You can also define a 3-category without picking $\mathcal{C}$, by treating the fusion categories as the objects, 1-morphisms being the bimodule categories, 2-morphisms functors of bimodule category and 3-morphisms natural transformations of above. Think of this as a truncated version of $MonCat_{bim}$ from before.
\end{definition}

\begin{definition}[Brauer-Picard 3-groupoid]
The objects are fusion categories Morita-equivalent to your starting $\mathcal{C}$. The 1-morphisms are the Morita-equivalences. These are $\mathcal{C}-\mathcal{C}$ bimodule categories. The 2-morphisms are the functors of bimodule categories that give equivalences. The 3-morphisms are natural transformations that give equivalences now taking equivalence classes. This is the connected version.
\end{definition}

\begin{definition}[Brauer-Picard 2-groupoid]
The objects are fusion categories Morita-equivalent to your starting $\mathcal{C}$. The 1-morphisms are the Morita-equivalences. These are $\mathcal{C}-\mathcal{C}$ bimodule categories. The 2-morphisms are the functors of bimodule categories that give equivalences.
\end{definition}

\begin{definition}[Brauer-Picard groupoid]
Replace each of those 1-morphisms by their equivalence classes.
\end{definition}

\begin{definition}[Brauer-Picard group]
Take the arrows only above one representative $\mathcal{C}$.
\end{definition}

Taking the classifying space of this 3-groupoid gives a homotopy space. Picking the connected component by the Morita equivalence class of $\mathcal{C}$ we get a 3-type. Its nontrivial homotopy groups are

\begin{itemize}
\setlength\itemsep{-1em}
\item $\pi_1$ is the ordinary group $BrPic(\mathcal{C})$.\\
\item $\pi_2$ is the isomorphism classes of invertible objects of the center. These are the invertible bulk excitations.\\
\item $\pi_3$ is $\mathbb{C}^*$ or more generally the $\mathbb{G}_m$ of the ground field.\\
\end{itemize}

\begin{theorem}[\cite{ENO10}]
The Brauer-Picard group of $\mathcal{C}$ is isomorphic as a group to braided auto-equivalences of the center.
\end{theorem}

\begin{remark}
Just knowing the braided autoequivalence of the center (the action on the anyons) forgets the trivalent junction for the defect fusion which are interesting equivalences.
\end{remark}

\begin{definition}[G Extension]
A G graded category $\mathcal{C}$ is a tensor category with a decomposition into $\mathcal{C}_g \; \forall g \in G$ such that each are full abelian subcategories and the tensor product takes $\mathcal{C}_g \times \mathcal{C}_h$ to $\mathcal{C}_{gh}$. \cite{LeTuraev} The trivial sector $\mathcal{C}_e$ is a full tensor subcategory and each other is a $\mathcal{C}_e$ bimodule category.
A G-extension of a fusion category $\mathcal{D}$ is a G-graded fusion category $\mathcal{C}$ such that $\mathcal{C}_e$ is equivalent to $\mathcal{D}$. If G is finite this is still finitely many simples. Though other structures like pivotal or spherical may not carry over from $\mathcal{D}$.
This is thought of as the original theory and the G-fluxes.
\end{definition}

\begin{theorem}[ENO10 Thm1.3 \cite{ENO10}].\\
\begin{itemize}
\setlength\itemsep{-1em}
\item Equivalence classes of G-extensions of $\mathcal{C}$ are given by homotopy classes of maps of classifying spaces $BG$ to $B \underline{\underline{BrPic(\mathcal{C})}}$\\
\item Alternatively this is parameterized by $c \; \; G \to BrPic(\mathcal{C})$ ordinary group homomorphism, an M belonging to a $H^2 ( G , \pi_2  )$ torsor ($\pi_2$ is given the structure of a G rep by $c$) and an $\alpha$ belonging to a $H^3 ( G , \pi_3 = \mathbb{C}^*  )$ torsor. Certain obstructions $o_3 (c) \in H^3 ( G ,\pi_2 )$ and $o_4 ( c, M) \in H^4 ( G , \pi_3=\mathbb{C}^*)$ need to vanish\\
\end{itemize}

\end{theorem}

So we see that what is important is $G$ as a 3-groupoid with one object, interesting 1 arrows and boring 2 and 3 arrows. Gauging a symmetry is presented explicitly as a higher group. \cite{RyanCurt}. Note that $G$ is often $\mathbb{Z}_2$ for parity or time reversal \cite{PinhasJordanSnyder}. As far as the definition goes the 3-groupoid is natural, but computationally we have easy access to only the fundamental groups. 

\begin{remark}
One may think of the problem of computing a generalized cohomology theory $h^\bullet (BG)$ where we also compute $[BG , X]$ or at least it's stablization. \cite{Charles}
\end{remark}

\subsection{Of $Vect_\mathbb{C}$ and $SuperVect_\mathbb{C}$}

In the case that both $\pi_1$ and $\pi_2$ were trivial, then we would have an Eilenberg-MacLane target which would mean that $[BG,K(\mathbb{C}^*,3)]$ could be calculated with group cohomology \cite{WenGroupCoh}. When $\mathcal{C}$ is $Vect_\mathbb{C}$ (the trivial theory), $\pi_1$ is trivial because that is the classical $Br(k) = \setof{e}$ for an algebraically closed field (this is also true for the maximal abelian extension of $\mathbb{Q}$) and the $\pi_2$ is trivial because the only invertible vector space is the one dimensional one. This is not the general case for other $\mathcal{C}$ so we must take care of any version of group (super)cohomology we compute \cite{WenSuper}. This is the first step that determines the possible gaugings. The next step is to use each such $G$ extension as a new TFT. This is done for each $\omega$ cohomology class.

If we use $SuperVect_\mathbb{C}$ on the point, we are only using it as a fusion category without braiding so replace it with $Rep_\mathbb{C} ( \mathbb{Z}_2 )$. Using the other associator gives the double semion model which we will return to. This reduces the problem to gauging the toric code. The braided autoequivalences of the center are $\mathbb{Z}_2$ implementing electromagnetic duality. The invertible objects are $\mathbb{Z}_2 \times \mathbb{Z}_2$, the third homotopy is $\mathbb{C}^*$ as usual.

\section{M\"uger style G equivariantization and quasi-trivial extensions}

This takes  a semisimple linear monoidal category with $G$ acting on it via autoequivalences. It then outputs a semisimple monoidal category having a full monoidal $Rep \; G$ subcategory. If $\mathcal{C}$ is G-braided G-crossed, then the $Rep \; G$ is a symmetric subcategory. \cite{MuegerApp}.

\begin{definition}[G-braided G-crossed]
A G-braided G-crossed fusion category is a fusion category equipped with a not necessarily faithful grading by G, an action of g by functors $T_g$ that take $\mathcal{C}_h \to \mathcal{C}_{ghg^{-1}}$ and G-braidings $X \otimes Y \simeq T_g (Y) \otimes X$ if $X$ is homogenous in degree g. The component with identity grading has an ordinary braiding.
\end{definition}

\subsection{Picard 2-groupoid}

\begin{definition}[$Modc( \mathcal{B} )$ ]
For a braided fusion category $\mathcal{B}$, we have the monoidal 2-category of left $\mathcal{B}$ module categories. The 1-morphisms are functors of module categories and the 2-morphisms are natural transformations of such functors.
\end{definition}

\begin{definition}[Picard categorical 2-group]

Only the invertibles in the above monoidal 2-category. This gives a particular sub full categorical 2-subgroup of the Brauer-Picard that our map from $BG$ may or may not factor through. At the group level, $Pic(\mathcal{C}) \simeq Aut^{br} ( Z ( \mathcal{C} ) , \mathcal{C} )$

\end{definition}

Taking the classifying space of $\underline{\underline{Pic}}(\mathcal{B})$ gives a homotopy space. The nontrivial homotopy groups of this are

\begin{itemize}
\setlength\itemsep{-1em}
\item $\pi_1$ is the ordinary group $Pic(\mathcal{B})$\\
\item $\pi_2$ is the group of isomorphism classes of invertible objects in $\mathcal{B}$\\
\item $\pi_3$ is $\mathbb{C}^*$ or more generally the $G_m$ of the ground field.\\
\end{itemize}

\begin{theorem}[ENO10 Thm\;7.12 \cite{ENO10}].\\
\begin{itemize}
\setlength\itemsep{-1em}
\item Equivalence classes of G-braided G-crossed categories with faithful G-grading having trivial component $\mathcal{B}$ are given by homotopy classes of maps of classifying spaces $BG$ to $B \underline{\underline{Pic(\mathcal{B})}}$\\
\item Alternatively this is parameterized by $c \; \; G \to Pic(\mathcal{B})$ ordinary group homomorphism, an M belonging to a $H^2 ( G , \pi_2  )$ torsor ($\pi_2$ is given the structure of a G rep by $c$) and an $\alpha$ belonging to a $H^3 ( G , \pi_3 = \mathbb{C}^*  )$ torsor. Certain obstructions $o_3 (c) \in H^3 ( G ,\pi_2 )$ and $o_4 ( c, M) \in H^4 ( G , \pi_3=\mathbb{C}^*)$ need to vanish\\
\end{itemize}

\end{theorem}

\subsection{De/Equivariantization}

\begin{definition}[Equivariantization \cite{MuegerApp}]
Let $\beta$ be an action of group $G$ on $\mathcal{C}$. Then $\mathcal{C}^G$ is the category whose objects are pairs $(X, \setof{u_g} )$ where $u_g :\; ^g X \to X$ is a system of isomorphisms such that all squares with $^{gh} X$ commute. The Hom sets are those $s \in Hom_{\mathcal{C}} ( X ,Y )$ such that the squares with $u_g$ and $v_g$ also commute for all $g$. Think of homotopy fixed points.
\end{definition}

\begin{proposition}[\cite{MuegerApp}]
If $\mathcal{C}$ is a braided G-crossed G-category, then $\mathcal{C}^G$ is braided ( not G-braided ).
\end{proposition}

\begin{theorem}[\cite{MuegerApp}]
If $\mathcal{C}$ is a braided fusion category, $\mathcal{S} \simeq Rep \; G$ is a full monoidal subcategory, $A$  is the corresponding commutative etale algebra object (think of the G-rep $Fun(G,k)$ with pointwise multiplication ). \\

Then the left $A$ modules in $\mathcal{C}$, called $_A \mathcal{C}$ is a braided G-crossed fusion category.\\
Then $(_A \mathcal{C})^G \simeq \mathcal{C}$ as a braided fusion category.\\
If $\mathcal{D}$ is a G-braided G-crossed fusion category, we can equivariantize it, then find $Rep G \subset \mathcal{D}^G$ and then de-equivariantize to $_A ( \mathcal{D}^G ) \simeq \mathcal{D}$ as a braided G-crossed fusion category.\\
\end{theorem}

\begin{proposition}[\cite{Mueger,TurionMO}]
Suppose $\mathcal{C}$ is a unitary fusion category, then the G de-equivariantization result is also unitary.
\end{proposition}

\begin{definition}[Modularization]
Take a ribbon category and consider a collection of invertible simple objects $\mathcal{G}$ which satisfy:

\begin{itemize}
\setlength\itemsep{-1em}
\item Closed under tensor product\\
\item Every object is transparent to all of $\mathcal{C}$\\
\item All dimension 1 rather than $-1$ ( invertibility only gave $\pm 1$)\\
\item The twist factors are all 1.\\
\end{itemize}

Taking the quotient with respect to this collection $\mathcal{C}//\mathcal{G}$ is the modularization. There is the essentially surjective functor $\mathcal{C} \to \mathcal{C}//\mathcal{G}$ because every object comes from $F(X) \; X \in Obj(\mathcal{C})$ but it is very much not full ( not surjective on the induced maps of hom sets )
 
\end{definition}

\begin{proposition}[\cite{LeTuraev} 1.9]
Suppose that G is finite and a neutral-modular G-category $\mathcal{C}$ is regular. Then $\mathcal{C}$ is a modular category in the ungraded sense as well. Without the regularity assumption this may fail because neutral-modular G-category only checks the nondegeneracy of the S matrix on the $g=e$ neutral component. (In \cite{LeTuraev} this is just called modular G-category)
\end{proposition}

\begin{proposition}[\cite{kirillov2004g} 10.2]
A unitary G-crossed G-braided fusion extension of a unitary modular category is modular ( distinguish with neutral-modular) if and only if the equivariantization is modular.
\end{proposition}

\begin{remark}
In \cite{AnyonCondensation}, Kong considers the problem of giving a modular category $\mathcal{C}$ and performing condensation via giving a subcategory $\mathcal{D}$ which is also modular. This is given by the data of a connected commutative etale algebra object. So de-equivariantization provides an example of the first step of anyon condensation. This is applied at the level of 3,2,1 extended field theories not fully extended theories. There is no guarantee of a spherical $\mathcal{X}$ such that $Z ( \mathcal{X} ) \simeq \mathcal{C}$ or $\mathcal{D}$.
\end{remark}

\begin{theorem}[\cite{gelaki09} 3.3,3.5]
Let $\mathcal{C}$ be a G extension of $\mathcal{D}$, then the relative center has the canonical structure of a braided G-crossed category. This then G equivariantizes to give the $(Z_{\mathcal{D}} (\mathcal{C} ))^G \simeq Z(\mathcal{C})$.
\end{theorem}

\begin{cor}[\cite{gelaki09} 3.7]
Let G be a finite group. A fusion category $\mathcal{A}$ is Morita equivalent to a G-extension of some $\mathcal{B}_e$ if and only if $Z( \mathcal{A} )$ contains a Tannakian subcategory $\mathcal{E} = Rep(G)$. This implies that we may see $\mathcal{A}$ is Morita equivalent to a G-graded fusion category $\bigoplus \mathcal{B}_g$ with $Z( \mathcal{B}_e ) \simeq \mathcal{E}^{'}_G$ as braided tensor categories. This gives a procedure for translating from a G-extension to braided G-crossed extension of the center. This is the difference in perspectives between \cite{ENO10,Mueger}. So if $\mathcal{D}$ is the original theory with particles $Z(\mathcal{D})$ and we gauge it to get $\mathcal{C}$ and then can condense.
\end{cor}

\begin{proof}
Suppose $\mathcal{A}$ is a G-extension $\bigoplus \mathcal{A}_g$. The Tannakian subcategory is given by for every representation $(V,\pi )$ of G giving the object $Y_\pi$ of $Z(\mathcal{A})$ defined as $V \otimes \mathbf{1}$ as an object of $\mathcal{A}$ and the half braiding $\pi(g) \otimes id_X \; \; X \otimes Y_\pi  \simeq V \otimes X \to V \otimes X \simeq Y_\pi \otimes X$ for $X \in \mathcal{A}_g$. The morphisms can also be checked as well as the converse.
\end{proof}

The M\"uger style procedure gives us $(Z_{\mathcal{D}} (\mathcal{C} ))^G$, but G-extension gives us $Z( \mathcal{C})$. This implies that even though we don't notice the difference down to codimension 2, we will notice the difference at codimension 3. For example, this can be accomplished by introducing an interval cut at some time to create point defects of the space-time at the endpoints of the interval at the instant of the introduction of the cut. Also note that the construction is conceptually simpler.

\begin{definition}[Nilpotent]
A fusion category $\mathcal{A}$ is called nilpotent if there is a sequence of finite groups $G_i$ and a sequence of fusion subcategories
$Vec \subset \mathcal{A}_1 \subset \cdots \mathcal{A}_n = \mathcal{A}$ such that each step is an extension by $G_i$. The smallest such n is called the nilpotency class of $\mathcal{A}$.
\end{definition}

\begin{definition}[Solvable]
A fusion category $\mathcal{A}$ is called solvable if there is a sequence of finite groups $G_i$ cyclic of prime order and a sequence of fusion subcategories
$Vec \subset \mathcal{A}_1 \subset \cdots \mathcal{A}_n = \mathcal{A}$ such that each step is an extension or equivariantization by $G_i$.
\end{definition}

\subsection{Other Sorts of Extension}

\begin{proposition}[\cite{ENO10} 7.2]
$\pi_1 B\underline{\underline{Out(\mathcal{D})}}$ is $Out(\mathcal{C})$. $\pi_2 B\underline{\underline{Out(\mathcal{D})}}$ is invertible objects of the center. $\pi_3$ is $\mathbb{C}^*$. 
\end{proposition}

\begin{proposition}[\cite{ENO10} 7.4]
$\pi_1 B\underline{Eq(\mathcal{D})}$ is $Eq(\mathcal{C})$. $\pi_2 B\underline{Eq(\mathcal{D})}$ is $Aut_\otimes (Id)$.
\end{proposition}

\begin{proposition}[Quasi-trivial extensions \cite{ENO10} 7.10]
Quasi-trivial extensions $\mathcal{C}$ ( where every $\mathcal{C}_g$ contains an invertible object ) of a fusion category $\mathcal{D}$ up to graded equivalence are in natural bijection with $BG \to B\underline{\underline{Out(\mathcal{D})}}$. The ones that are actually trivial ( $Vec_G \ltimes \mathcal{D}$ ) can be factored through $BG \to B\underline{Eq(\mathcal{D})} \to  B\underline{\underline{Out(\mathcal{D})}}$. $\underline{\underline{Out(\mathcal{D})}}$ is a 2-subgroup of $\underline{\underline{BrPic(\mathcal{D})}}$ by only including the bimodule categories that are $\mathcal{D}$ as left module categories and the right module category structure is twisted by some autoequivalence which is determined up to conjugation. The higher structure must also respect this restriction.
\end{proposition}

In summary, we have trivial extensions determined by $B\underline{Eq}$ which also give quasi-trivial extensions determined by maps to $B\underline{\underline{Out}}$. There are examples of G-extensions determined by maps to $B\underline{\underline{BrPic}}$. G-braided faithful G-crossed extensions also give a special kind of G-extension when the map factors through $B\underline{\underline{Pic}}$.

\section{Hopf Algebra Transmutation}

When we actually have a forgetful functor, we may may talk about Hopf algebras instead of their representation categories. For example, we may write the equivariantization as $Rep ( H \rtimes kG ) \simeq Rep (H)^G$ \cite{Majid,galindo}. 

If $H$ is a co-quasitriangular Hopf algebra, then we get a module category for $D(A,r)-mod$ by taking representations of the covariantized (braided form) algebra $A_r$ by the $A_r \to D(A) \otimes A_r$ comodule algebra structure  \cite{Yinhuo} . The Heisenberg double shows up in this way \cite{BrochierJordan}. We could also consider the coideal subalgebra case $A_q ( K \backslash G / K ) \to A_q (G)$ \cite{Noumi}. $A_q ( K \backslash G / K )$ is intimately related with Macdonald and Koornwinder polynomials. However there $q$ is generic so this does not do the semisimplification procedure that is needed for our setting. This procedure can be used to construct bimodule categories for related categories, but not precisely what we need here \cite{Impurities}. This semisimplicity concerns would need to be addressed if one were to construct these would be ``Turaev-Viro-Macdonald type" theories.

The restricted quantum group as it appears in logarithmic conformal field theory  gives a Hopf algebra, but sadly ruins semisimplicity \cite{FGST}. But even without semisimplicity, a co-quasitriangular Hopf algebra can be covariantized and its representations applied for \cite{Blanchet,Lyubashenko} type field theories rather than TQFT's of Lurie type. This is the dual of the quasitriangular $\bar{D}$ in \cite{FGST}. This perspective for defects in LCFT is work in progress.

\begin{theorem}[3 of \cite{Yinhuo}]
Let $\mathcal{C}$ be a braided fusion category. Then the Drinfeld center of $\mathcal{C}$ is equivalent to the category of finite dimensional left comodules over some braided Hopf algebra $_R H_\mathcal{C}$ (Hopf algebra object). Moreover, if $A$ is a braided bi-Galois object over $_R H_\mathcal{C}$, then the cotensor functor $A \Box -$ defines a braided autoequivalence of the Drinfeld center of $\mathcal{C}$ (  trivializable on $\mathcal{C}$ ) if and only if $A$ is quantum commutative. 
\end{theorem}

The braided autoequivalences that are not trivializable over $\mathcal{C}$ are not of the form $A \Box -$ for any braided bi-Galois objects over $_R H_\mathcal{C}$. In the case of $\mathcal{C}$ being the representation category of a semisimple Hopf algebra over algebraically closed field, then this gives all of $Aut^{br} ( Z ( \mathcal{C} ) , \mathcal{C} )$ which is the image of $Pic \subset BrPic$.

\section{Sphericality and Anomalies}

The question is knowing that we start with the adjectives unitary and spherical, and perform G-extension we want to be able to stay unitary and spherical. That is we must take the G-graded fusion category result, forget the G grading and see that as a spherical or pivotal fusion category. If not, we can only define a framed theory or a combed field theory depending on what structure is kept. This is the main question of the subject. 

\begin{remark}
If the theory was defined after giving a 2-framing ( Trivialization of $TX \bigoplus TX$ ) we could understand that as a bounding 4-manifold providing anomaly inflow. In that case you define your theory as a relative theory \cite{FreedSRE,RyanCurt} such as in the Crane-Yetter context.
\end{remark}

\begin{conjecture}[Etingof, Nikshcych, Ostrik \cite{OnFusion} 2.8]
Every fusion category admits a pivotal structure.
\end{conjecture}

\begin{cor}
This implies that even after $G$ extension, we can get a combed field theory after a choice. The ENO conjecture is proven for representation categories of semisimple quasi-Hopf algebras, but again unfortunately we have lost our fiber functor so we cannot use it as a theorem.
\end{cor}

\section{Quantum Group Example}

The Turaev Viro model was defined originally as a state sum model, but further efforts have realized it in the cobordism framework. Assuming \cite{DSPS}, this is given by assigning a spherical fusion category to the framed point. Knowing that and the cobordism theorem is enough to determine the entire theory \cite{Balsam1,Balsam2,Balsam3}. Common sources for fusion categories are either finite groups or quantum groups at roots of unity \cite{ReshetikhinTuraev}. In terms of the Brauer-Picard and equivariantization, the case of finite abelian groups is in the original \cite{ENO10,KapustinSaulina}, representation categories for finite groups \cite{FiniteRepCatCase,Davydov} and the Asaeda-Haagerup subfactor example is computed as well in \cite{PinhasSnyder}. The quantum groups seem to be missing from the literature. See the appendix for the general construction as well as the explicit data and notations for the examples considered.

Because the category in this case is modular, the center is equivalent to the product $\mathcal{C} \boxtimes \mathcal{C}^{op}$ so the Brauer-Picard group is the same as braided autoequivalences of $Z(\mathcal{C}) \simeq \mathcal{C} \boxtimes \mathcal{C}^{op}$. These possibilities are checked by combinatorial criteria for where to send generating objects. The next step is determined by automorphisms of some planar algebras \cite{NoahMO}. Conveniently many of the models of physical interest are in the rank 1 case where the calculation is significantly easier because $Aut(Temperley-Lieb)=\setof{id}$.

\begin{remark}
A Mathematica notebook for the combinatorial criteria calculation can be made available upon contact.
\end{remark}

\subsection{Double Semion}

The braided autoequivalences of the center are given by $1 \boxtimes s$ and $s \boxtimes 1$ switching. This is the combinatorial step. Then give the functors that have that as the underlying action on objects. In this case, there is nothing else to check so $\pi_1 = \mathbb{Z}_2$. The group of invertible objects of the center form a $\mathbb{Z}_2 \times \mathbb{Z}_2$. The action of $\pi_1$ on $\pi_2$ is given by switching the factors.

\subsection{Double Fibonnacci}

The braided autoequivalences of the center are given by $1 \boxtimes F$ and $F \boxtimes 1$ switching. This is the combinatorial step. Then give the functors that have that as the underlying action on objects. In this case, there is nothing else to check so $\pi_1 = \mathbb{Z}_2$. The group of invertible objects of the center form a trivial group.

\subsection{$A_1$ General $\ell$}

The classes of simple objects will be parameterized by pairs in $[0,\ell-2]*\omega$. We need to see where the objects $(\omega ,0)$ and $( 0 , \omega )$ go. These can go to various other simple objects depending on $\ell$.

\begin{tabular}{c|c|c}
\centering
obj & dim & twist \\
\hline
$( \omega , 0 )$ & $[1+1]*[1]$ & $q^{3/2}*1$ \\
$( 0 , \omega )$ & $[1]*[1+1]$ & $1*q^{3/2}$ \\
$ ( (\ell -3 )*\omega , 0 )$& $[1+\ell-3]*[1]$ & $q^{1/2*(\ell-3)(\ell-1)}*1$ \\
$ ( 0, (\ell -3 )*\omega  )$& $[1]*[1+\ell-3]$ & $1*q^{1/2*(\ell-3)(\ell-1)}$ \\
$( \omega , (\ell-2) \omega )$& $[1+1]*[1+\ell-2]$& $q^{3/2}*q^{1/2*(\ell-2)(\ell)}$ \\
$( (\ell-2)\omega , \omega )$ & $[1+\ell - 2]*[1+1]$ & $q^{1/2*(\ell-2)(\ell)}*q^{3/2} $ \\
$ ( (\ell -3 )*\omega , (\ell-2) \omega )$& $[1+\ell-3]*[1+\ell - 2]$ & $q^{1/2*(\ell-3)(\ell-1)}*q^{1/2*(\ell-2)(\ell)}$ \\
$ ( (\ell -2 ) \omega, (\ell -3 )*\omega  )$& $[1+\ell -2 ]*[1+ \ell -3 ]$ & $q^{1/2*(\ell-2)(\ell)}*q^{1/2*(\ell-3)(\ell-1)}$ \\
\end{tabular}

\begin{eqnarray*}
q^{3/2} == q^{ 1/2*(\ell-3)(\ell-1) } &\implies& \ell = 0 \mod 4\\
1 == q^{1/2 * (\ell-2)*\ell} &\implies& \ell = 2 \mod 4\\
q^{3/2} == q^{1/2*(\ell-2)(\ell)}*q^{1/2*(\ell-3)(\ell-1)} &\implies& \ell = 1 \mod 2\\
\end{eqnarray*}

\subsubsection{$\ell = 0 \mod 4$}

Based on twists and dimensions $(\omega, 0)$ and $(0, \omega)$ can go to $((\ell-3)\omega , 0)$ and $(0,(\ell-3)\omega)$ or $(\omega,0)$ and $(0,\omega)$.

The higher $\pi$ are $\pi_2=\mathbb{Z}_2 \times \mathbb{Z}_2$ from the classes of $(\ell-2,0)$ and $(0,\ell-2)$ and $\pi_3=\mathbb{C}^*$.

\subsubsection{$\ell = 1 \mod 4$}

Based on twists and dimensions $(\omega, 0)$ and $(0, \omega)$ can go to $((\ell-3)\omega , (\ell-2)\omega )$ and $((\ell-2) \omega,(\ell-3)\omega)$ or $(\omega,0)$ and $(0,\omega)$

The higher $\pi$ are again $\mathbb{Z}_2 \times \mathbb{Z}_2$ and $\mathbb{C}^*$.

\subsubsection{$\ell = 2 \mod 4$}

Based on twists and dimensions $(\omega, 0)$ and $(0, \omega)$ can go to $((\ell-2)\omega , \omega)$ and $( \omega,(\ell-2)\omega)$ or $(\omega,0)$ and $(0,\omega)$

The same $\pi_2$ and $\pi_3$ remain.

\subsubsection{$\ell = 3 \mod 4$}

Based on twists and dimensions $(\omega, 0)$ and $(0, \omega)$ can go to $((\ell-3)\omega , (\ell-2)\omega )$ and $((\ell-2) \omega,(\ell-3)\omega)$ or $(\omega,0)$ and $(0,\omega)$

The same $\pi_2$ and $\pi_3$ remain.

\section{Conclusion}

This gives an overview of how the cobordism hypothesis and G-extensions fit together to give the act of taking symmetry defects and permeating them in order to gauge finite group symmetries. There is relation with the procedure of equivariantization when extra braiding data is given, but they give slight differences on extremely stratified spacetimes. These are potentially implementable as certain kinds of quenches because we already get codimension 2 when considering a system with boundary and then crossing with the half open interval of the future. This illustrates the difference between establishing the equality of the theories at the level of partition functions or all the way down with duality walls to say what it means to be in the same (physical) phase. We mentioned the relationship with Majid's theory of transmutation though that is only applicable in the examples when the weak Hopf algebra is actually an honest Hopf algebra which are not the categories considered here.

In the case of Walker-Wang models, the story can be described as a $(\infty , 4)$ version with which we can repeat the procedure of finding fully dualizable objects and asking what happens upon changing them with the data of a finite group. This is a much simpler case than most $3+1$ dimensional systems because this $(\infty , 4)$ category is only as interesting as the $(\infty ,3)$ one discussed here (which is still extremely interesting ). This means that instead of mapping $BG$ into a homotopy 3-type we map into a homotopy 4-type with $\pi_1 = 0$. Again if we are just using a version of $Vect_\mathbb{C}$ shifted up we get a group cohomology classification. We can bump this game arbitrarily, but this only uses the aspects that come from 3 dimensions vs the new instruments that enter as the conductor signals a change in dimension. Proceding with this analogy matches up with \cite{Charles}.

\section{Appendix}

\subsection{Quantum Groups at Roots of Unity}

For the standard quantum group $U_q \mathfrak{g}$ with $\mathfrak{g}$ a ABCDEFG type Lie algebra and $q^2$ a primitive $\ell$'th root of unity $\ell = k + h^\vee$ with k positive. The associated quotient category of tiltings by negligibles has simples labelled by the dominant weights in the alcove. This is explained in many sources such as \cite{Rowell,RowellStongWang,Sawin,ChariPressley}

\begin{eqnarray*}
\langle \lambda + \rho \mid \theta_0 \rangle < \ell \; \; \text{if}\; \; m \mid \ell\\
\langle \lambda + \rho \mid \theta_1 \rangle < \ell \; \; \text{if} \; \; m \not{\mid} \ell\\
\end{eqnarray*}

where $m$ is the ratio of long to short root square lengths, $\theta_0$ is the highest root, and $\theta_1$ is the highest short root. The inner product is the one normalized to $2$ for short roots.

\begin{eqnarray*}
S_{\lambda , \mu} &=& \frac{ \sum \epsilon (w) q^{2 \langle \lambda + \rho \mid w ( \mu + \rho ) \rangle }}{\sum \epsilon (w) q^{2 \langle \rho \mid w ( \rho ) \rangle }}\\
\theta_\lambda &=& q^{\langle \lambda \mid \lambda + 2 \rho \rangle}\\
d_\lambda &=& \prod_{\alpha} \frac{ [\langle \lambda + \rho \mid \alpha \rangle ] }{ [ \langle \rho \mid \alpha \rangle ] }
\end{eqnarray*}

\begin{example}[Semion model \cite{RowellStongWang}]
$A_1$ at $q=e^{\pi i /3}$, alternatively $k=1$ for the loop group using $k+h=3$. Also realized in $E_7$ level 1.

\begin{center}
\begin{itemize}
Objects: $1$ and $s$\\
Fusion algebra: $s^2 = 1$\\
Quantum dimensions: Both $1$ so $D = \sqrt{2}$\\
Twists: $1$ and $i$\\
Central charge: $1$\\
Nontrivial Braidings $R^{ss}_1 = i$\\
S matrix:   $\frac{1}{\sqrt{2}} \begin{pmatrix}1&1\\1&-1\end{pmatrix}$\\
Associator matrix: $F_s^{s,s,s} = -1$ with intermediate having to be $1$\\
\end{itemize}
\end{center}

\end{example}

\begin{example}[Fibonacci model \cite{RowellStongWang}]
$A_1$ at $q=e^{\pi i /5}$, alternatively $k=3$ then only take integer highest weights and not of the half integers.

\begin{center}
\begin{itemize}
Objects: $1$ and $F$\\
Fusion algebra: $F^2 = 1+F$\\
Quantum dimensions: $1$ and $1-\phi$ so $D = 3-\phi$\\
Twists: $1$ and $u^{-2}$ where $u$ is a primitive $10$'th root of unity solution of $u^2-\phi*u+1=0$\\
Central charge: $u^{-1/2}$\\
Nontrivial Braidings $R^{FF}$ $F \otimes F \simeq 1 \bigoplus F \to 1 \bigoplus F \simeq F \otimes F = 
\begin{pmatrix}u^2&0\\
0&u\\\end{pmatrix}$\\
Associator matrix: $F_F^{F,F,F} = \begin{pmatrix}-\phi&-\phi\\
1&\phi\\\end{pmatrix}$ using $1$ and $F$ as the basis for the intermediate space.\\
\end{itemize}
\end{center}

\end{example}

\begin{example}[General $A_1$]
$A_1$ at $q=e^{\pi i / \ell}$
\begin{center}
\begin{itemize}
Objects: $\langle \lambda + \rho \mid \theta_0 \rangle \lt \ell$. So $\lambda \in [0,\ell-2]*\omega_1 =x*\omega_1$\\
Quantum dimensions:  $\prod_{\alpha} \frac{ [\langle \lambda + \rho \mid \alpha \rangle ] }{ [ \langle \rho \mid \alpha \rangle ] } = \frac{[1+x]}{[1]} = [1+x]$\\
\end{itemize}
\end{center}

For the original Hopf algebra

\begin{eqnarray*}
R &=& q^{H \otimes H/4} \sum_{n=0}^\infty \frac{(1-q^{-1})}{[n]!} q^{n(1-n)/4} ( q^{nH/4} (L^+)^n ) \otimes ( q^{-nH/4} (L^-)^n )\\
\rho_{\Lambda_1} \otimes \rho_{\Lambda_2}( R) \ket{j_1 , m_1} \otimes \ket{j_2 , m_2} &=& \sum_{n\geq 0} \sqrt{ \binom{j_1-m_1}{n}_q \binom{j_2 + m_2}{n}_q \frac{[j_1+m_1+n]! [j_2 - m_2 + n]!}{ [j_1+m_1]! [j_2-m_2]!}} \\&*&q^{n(1-n)/4} q^{1/2( m_2 n -m_1 n + 2m_1m_2)} (1-q^{-1})^n \ket{j_1 , m_1 + n} \ket{j_2 , m_2 - n}\\
\end{eqnarray*}

The braiding for the semisimplification comes from this braiding via Mac Lane's quotient construction.

\end{example}

\begin{definition}[Categorically Morita equivalent]
$\mathcal{A}$ and $\mathcal{B}$ fusion categories are categorically Morita equivalent if there is an $\mathcal{A}$ module category $\mathcal{M}$ such that $\mathcal{B} \simeq (\mathcal{A}_{\mathcal{M}}^*)^{op}$. For example, $Rep(H)_{Vec}^* \simeq Rep(H^*)$ giving a categorical Morita equivalence with the Hopf dual's representations.
\end{definition}

\begin{theorem}[\cite{NikshychReview} Thm 5.1]
Two fusion categories are categorically Morita equivalent if and only if their Drinfeld centers are equivalent as braided fusion categories.
\end{theorem}

\begin{definition}[Group Theoretical]
A fusion category is called group theoretical if it is categorically Morita equivalent to a pointed fusion category. Because all pointed fusion categories are equivalent to $Vec_G^\omega$ for some finite group and 3-cocycle, this is equivalent to asking for $ ((Vec_G^\omega)_{\mathcal{M}}^*)^{op}$
\end{definition}

\begin{definition}[Weakly group theoretical]
These are the ones that art categorically Morita equivalent to a nilpotent fusion category. This class is closed under all extensions and equivariantizations. \cite[Prop 3.26]{NikshychReview}
\end{definition}

\begin{definition}[Solvable]
This is the class of fusion categories categorically Morita equivalent to a cyclically nipotent fusion category. This class is closed under extension or equivariantization by solvable groups, Morita equivalence, tensor products, subcategories and components of quotients.
\end{definition}

\begin{theorem}[Joyal Street]
Pointed braided fusion categories $Vec_A^\omega$ are equivalent to giving an abelian group and a quadratic form. Braided tensor functors correspond to quadratic form preserving homomorphisms of abelian groups. The cocycle $\omega$ is given in terms of the quadratic form.
\end{theorem}

Consider the cases when the framed point is assigned the category $Vec_A$ for some abelian group $A$. We still have bimodule categories etc between them. We may truncate this to give a 1-category $Bimod_{ab}$ where the bimodule categories are now up to equivalence only.

\begin{definition}[Metric group]
Let $E$ be a finite abelian group. A bicharacter is a biadditive map to $\mathbb{C}^*$. A symmetric bilinear form is a symmetric bicharacter. A quadratic form is a function $E \to \mathbb{C}^*$ such that $q(x)=q(x^{-1})$. This gives a symmetric bilinear form by $\frac{q(xy)}{q(x)q(y)}$ or $\frac{q(xy)q(e)}{q(x)q(y)}$. Taking $log$ we see $\log q(xy) - \log q(x) - \log q(y) \bigg(+ \log q(e) \bigg)$ so the $q$ in \cite{BelovMoore,KapustinSaulina} is $\frac{1}{2\pi i} \log q$.
\end{definition}

\begin{remark}
We will see that the quadratic refinement of the symmetric bilinear form is what is important in the Brauer Picard rather than the symmetric bilinear form itself.
\end{remark}

\begin{definition}[Lagrangian]
A subgroup L is isotropic if $q(a)=1 \; \forall a \in L$ and Lagrangian if $\abs{L}^2 = \abs{E}$ giving maximal isotropic. Think of the order of the group as $e^{\# n}$ for a symplectic space of dimension $n$.
\end{definition}

\begin{definition}[Lagrangian Corr]
This particular version of the category of Lagrangian correspondences is objects are metric groups and morphisms $(E_1 , q_1 ) \to (E_2 , q_2 )$ are formal $\mathbb{Z}_+$ combinations of Lagrangian subgroups in $(E_1 \bigoplus E_2 , q_1^{-1} \bigoplus q_2 )$. Composition for Lagrangian subgroups is $ M L = m(M,L) M \otimes L$ where $\otimes$ is the Lagrangian spanned by all $a_1 , a_3$ such that $\exists a_2$ with $a_1 , a_2$ in $L$ and $a_2 , a_3$ in $M$ and $m(M,L)$ counts the number of $a_2$. This is then extended linearly for formal combinations.
\end{definition}

\begin{remark}
Note that rescalings of the quadratic form by $\mathbb{C}^*$ leave Lagragian condition alone so $(E_1 , q_1 ) \to (E_2 , q_2 )$ can be identified exactly with $(E_1 , r*q_1 ) \to (E_2 , r*q_2 )$. A way to fix this ambiguity in the objects of $Lag$ is to require a Gauss-Milgram constraint for every object using some $\sigma \in \mathbb{C}/(8 * \mathbb{Z})$

\begin{eqnarray*}
\abs{D}^{-1/2} \sum q( \gamma ) &=& e^{2 \pi i \sigma / 8}\\
\end{eqnarray*}

It also leaves the bilinear form invariant in the more symmetric definition. We get a full subcategory by picking only these objects, but it is more natural to see this as a quotient rather than a sub. This is in the sense that we have a free action on the objects and the action on morphisms is trivial so the quotient category makes sense.

\end{remark}

\begin{remark}
We can take a subcategory where the quadratic form only gives roots of unity. In that case we have a $\mathbb{Q}/\mathbb{Z}$ valued quadratic form. Now only rescalings by roots of unity are allowed. Now $e^{2 \pi i \sigma / 8} \in \mathbb{Q}^{ab}$ (The abelian Galois extension). Call this $Lag_{unity}$
\end{remark}

\begin{proposition}[\cite{ENO10} 10.2]
The groupoid of isomorphisms in $Lag$ is naturally isomorphic to the groupoid of isometries of metric groups. This is because the invertible Lagrangians will be the graphs of isomorphisms in order to be transverse to both $E_1$ and $E_2$. 
\end{proposition}

\begin{definition}
$Lag_{hyp}$ is the full subcategory of $Lag$ where only objects of the form $(\mathscr{D}:=A \bigoplus A^*, ev)$ are allowed and the same morphisms. This is not preserved under the rescaling autoequivalence, but it is included into $Lag_{unity}$\\
\end{definition}

\begin{theorem}[\cite{ENO10} 10.4,10.5]
$Bimod_{ab}$ and $Lag_{hyp}$ are equivalent categories. This equivalence together with 10.2 gives that $BrPic(Vec_A) = O(A \bigoplus A^*)$. The equivalence sends $Vec_A$ to $(A \bigoplus A^* , ev)$.
\end{theorem}

\begin{definition}[Witt group $\mathbf{W}$]
The Witt group of braided fusion categories which is given by the equivalence relation $[ \mathcal{C}_1 ] \equiv [ \mathcal{C}_2 ]$ when $\mathcal{C}_1 \boxtimes Z(\mathcal{A}_1) \simeq \mathcal{C}_2 \boxtimes Z ( \mathcal{A}_2 )$. You are allowed to attach as much of the trivial topological orders $Z ( \mathcal{A} )$ to either side in order to get equivalence \cite{DNO,DMNO,DGNO,FidKitaev}. From this perspective, all field theories we have described are trivial as the only categories considered as braided are Drinfeld centers.
\end{definition}

\begin{theorem}[\cite{DMNO} 6.4]
$\mathbf{W}$ contains a cyclic subgroup of order $16$ generated by the classes of Ising braided fusion categories. The $\mathbb{Z}_8$ subgroup is given by the classes of $C(A,q)$ where $A$ is an order 4 metric group such that $\exists u \; q(u)=-1$. $\mathbf{W}$ also contains a subgroup corresponding to non-degenerate pointed braided fusion categories which is equivalent to the Witt group of metric groups $ = \bigoplus_{p} W_{ptd} (p)$ where for $W_{ptd} (2) = \mathbb{Z}_8 \bigoplus \mathbb{Z}_2$, $W_{ptd} (3 \; \text{mod} \; 4 ) = \mathbb{Z}_4$ and $W_{ptd} (1 \; \text{mod} \; 4 ) = \mathbb{Z}_2 \bigoplus \mathbb{Z}_2$
\end{theorem}

In order to reconstruct the full structures of $B\underline{\underline{BrPic}}$ or $B\underline{\underline{Pic}}$ etc, we must now about the actions of $\pi_1$ on $\pi_\bullet$ and Postnikov k-invariants since we need to reconstruct the topological space from only this homotopy group data. The homotopy calculations come from mapping spaces into the Postnikov tower.

\begin{definition}[Postnikov Tower]
A Postnikov tower is a system of path-connected spaces $X_n \cdots X_1 \to X_0$ with each map being a fibration. $X_n$ matching the desired homotopy groups for all $k \leq n$ and vanishing above that. The space is then reconstructed as the inverse limit. Each fiber is an Eilenberg MacLane space $K ( \pi_n , n )$.
\end{definition}

Algebraic models for this are also available in the same manner as crossed modules for 2-types.

\bibliographystyle{ieeetr}
\bibliography{GEqSPT}

\end{document}